\documentclass[11pt,twoside]{article} 
\usepackage{acta-info}
\usepackage{euler}

\newcommand{\vol}{\textbf{2,} 2 (2010) 194--209}   
\newcommand{\amss}{\amssubj{%
68Q45 
}}     
\newcommand{\CRs}{\CRsubj{%
F.4.3 
}}   
\newcommand{\keyw}{\keywords{%
context-free languages, linear languages,
pumping lemma, derivation tree, regular languages
}}   

\begin{document}

\title{
Pumping lemmas for linear and nonlinear context-free languages
}
\maketitle

\centerline{\emph{Dedicated to P\'al D\"om\"osi on his 65th birthday}}

\twoauthors{%
\href{http://www.inf.unideb.hu/~geza}{G\'eza Horv\'ath}
}{%
\href{http://www.unideb.hu/portal/en}{University of Debrecen} 
}{%
\href{mailto:geza@inf.unideb.hu}{geza@inf.unideb.hu} 
}{%
\href{http://www.inf.unideb.hu/~nbenedek}{Benedek Nagy} 
}{%
\href{http://www.unideb.hu/portal/en}{University of Debrecen}
}{%
\href{mailto:nbenedek@inf.unideb.hu}{nbenedek@inf.unideb.hu} 
}

\short{G. Horv\'ath, B. Nagy}{%
Pumping lemmas for linear and nonlinear 
languages
} 

\begin{abstract}
Pumping lemmas are created to prove that given languages
are not belong to certain language classes. There are several known pumping lemmas
for the whole class and some special classes of the context-free languages.
In this paper
we prove new, interesting pumping lemmas for special linear and context-free
language classes. Some of them can be used to pump regular languages in two place
simultaneously. Other lemma can be used to pump context-free languages in arbitrary
many places.
\end{abstract}

\section{Introduction}

The formal language theory and generative grammars form one of the basics of the field of theoretical computer
 science \cite{hopul,handbo}.
 Pumping lemmas play important role in formal language theory \cite{BARH,pali-pump}.
One can prove that a language does not belong to a given language class.
There are well-known pumping lemmas, for example, for regular and context-free languages.
The first and most basic pumping lemma is introduced by Bar-Hillel, Perles, and
Shamir in 1961 for context-free languages \cite{BARH}.
Since that time many pumping lemmas are introduced for various
language classes. Some of them are easy to use/prove, some of them
are more complicated. Sometimes a new pumping lemma is
introduced to prove that a special language does not belong to a
given language class.
Several subclasses of context-free languages are known, such as 
deterministic context-free and linear languages.
 The linear language class is strictly between the regular and the context-free ones.
In linear grammars only the
following types of  rules can be used: $A \to w$, $A \to uBv$ ($A,B$ are non-terminals, $w,u,v \in V^*$). In the sixties,
Amar and Putzolu defined and analysed a special subclass of linear languages, the so-called even-linear ones,
in which the rules has a kind of symmetric shape \cite{am-put1} (in a rule of shape $A\to uBv$, i.e., with non-terminal at the right hand side, the length of $u$ must equal to the length of $v$).
The even-linear languages are intensively studied, for instance, they play special importance in 
learning theory \cite{even-lin}. 
In \cite{am-put2} Amar and Putzolu extended the definition to any fix-rated linear languages.
They defined the $k$-rated linear grammars and languages, in which
the ratio of the lengths of $v$ and $u$ equals to a fixed non-negative rational number $k$ for all rules of the grammar containing non-terminal in the right-hand-side.
 They used the term
$k$-linear for the grammar class and $k$-regular for the generated language class.
In the literature the $k$-linear grammars and languages are frequently 
used for the metalinear grammars and languages \cite{hopul}, as they are extensions of
the linear ones (having at most $k$ nonterminals in the sentential forms). 
 Therefore, for clarity, we prefer the term fix-rated ($k$-rated) linear for those restricted linear grammars and languages
 that are introduced in \cite{am-put2}.
The classes $k$-rated linear languages are strictly between the linear and regular
 ones for any rational value of $k$. Moreover their union the set of all fixed-linear
 languages is also strictly included in the class of linear languages. 
 %
In special case $k=1$ the even-linear grammars and languages are obtained; while the case $k=0$ corresponds to 
the regular grammars and languages.
The derivation-trees of the $k$-rated linear grammars form pine tree shapes.
In this paper we investigate pumping lemmas for these 
languages also.
These new pumping lemmas work for regular languages as well, since every regular language
is $k$-rated linear for every non-negative rational $k$. In this way the words of a regular
language can be pumped in two places in a parallel way.
There are also extensions of linear grammars. A context-free grammar is said to be
$k$-linear if it has the form of a linear grammar
plus one additional rule of the form $S\rightarrow S_1S_2\ldots S_k$,
where none of the symbols $S_i$ may appear on the right-hand side of any
other rule, and $S$ may not appear in any other rule at all.
A language is said to be $k$-linear if it can be generated by a $k$-linear grammar,
and a language is said to be metalinear if it is $k$-linear for some positive integer $k$.
The metalinear language family is strictly between the linear and context-free
ones.
In this paper we also introduce a pumping lemma for not metalinear context-free
languages, which can be used to prove that the given language belongs to the class of
the metalinear languages.



\section{Preliminaries}

In this section we give some basic concepts and fix our notation.
Let $\mathbb N$ denote the non-negative integers and $\mathbb Q$ denote the
non-negative rationals through the paper.

A grammar is an ordered quadruple $G= ( N,V,S,H )$, where $N,V$
are the non-terminal and terminal alphabets. 
 $S \in N$ is the 
  initial letter. $H$ is a finite set of derivation rules. A rule is a pair
 written in the form $v \rightarrow w$ with $v\in (N
\cup V)^* N (N \cup V)^*$ and $w \in  (N \cup V)^*$.

Let $G$ be a grammar and $v,w \in (N \cup V)^*$. Then $v \Rightarrow w$ 
is a direct derivation if and only if there exist $v_1,v_2,v',w'
\in (N \cup V)^*$ such that $v = v_1 v' v_2$, $w = v_1 w' v_2$
and $v' \rightarrow w' \in H$. The transitive and reflexive closure of $\Rightarrow$ is
  denoted by $\Rightarrow^*$. 

The language generated by a grammar $G$ is 
$L(G) = \{ w | S
\Rightarrow^* w \wedge w \in V^* \}$.
Two grammars are equivalent if they generate the same
language modulo the empty word ($\lambda$). (From now on we do not care whether $\lambda  \in L$ or not.)

Depending on the possible structures of the derivation rules we
are interested in the following classes \cite{am-put2,hopul}.

\noindent {$\bullet$}  type 1, or context-sensitive (CS) grammars: for every
rule the next scheme holds: $uAv \rightarrow uwv$ with $A \in N$ and
$u,v,w \in (N \cup V)^*, w \ne \lambda$.

\noindent {$\bullet$}  type 2, or context-free (CF) grammars: for every
rule the next scheme holds: $A \rightarrow v$ with $A \in N$ and
$v \in (N \cup V)^*$.

\noindent {$\bullet$}  linear (Lin) grammars: each rule is one of the
next forms: $A\rightarrow v$, $A\rightarrow vBw$; where $A,B \in
N$ and $v,w \in V^*$.

\noindent {$\bullet$}  $k$-linear ($k$-Lin) grammars: it is a linear grammar
plus one additional rule of the form $S\rightarrow S_1S_2\ldots S_k$,
where $S_1,S_2,\ldots ,S_k\in N$, and none of the $S_i$ may appear on the
right-hand side of any other rule, and $S$ may not appear in any other rule at all.

\noindent {$\bullet$}  metalinear (Meta) grammars: A grammar is said to be
metalinear if it is $k$-linear for some positive integer $k$.

\noindent {$\bullet$}  $k$-rated linear ($k$-rLin) grammars: it is a linear grammar with the following property: there exists a rational number $k$ such that
for each rule of the form: $A\rightarrow vBw$: $\frac{|w|}{|v|} = k$ (where
$|v|$ denotes the length of $v$).

Specially with $k=1$:
\\
\noindent {$\bullet$}  even-linear (1-rLin) grammars. 

Specially with $k=0$: 

\noindent {$\bullet$}  type 3, or regular (Reg) 
 grammars: each
derivation rule is one of the following forms: $A\rightarrow w$,
$A\rightarrow wB$; where $A,B \in N$ and $w \in V^*$.


The language family regular/linear etc. contains all languages that can be generated by
regular/linear etc. grammars. 
We call a language $L$ fix-rated linear if there is a $k\in\mathbb Q$ such that
$L$ is $k$-rated linear. So the class of fix-rated linear languages includes all the
$k$-rated linear language families.
Moreover it is known by \cite{am-put2}, that for any value of $k\in\mathbb Q$
all regular languages are $k$-rated linear.

The hierarchy of the considered language classes can be seen in Fig. \ref{hier}.

\begin{figure}
\begin{picture}(300,100)(-70,-52)
\put (0,0){\oval(95,95)}
\put (0,0){\oval(75,75)}
\put (0,0){\oval(55,55)}
\put (0,0){\oval(35,35)}
\put (0,0){\oval(15,15)}
\put (30,40){\line(1,0){90}}
\put (24,29){\line(1,0){96}}
\put (16,18){\line(1,0){104}}
\put (8,7){\line(1,0){112}}
\put (0,-4){\line(1,0){120}}
\put(178,40){\makebox(0,0){\it Context-free languages}}
\put(174,29){\makebox(0,0){\it Metalinear languages}}
\put(164,18){\makebox(0,0){\it Linear languages}}
\put(185,7){\makebox(0,0){\it  Fix-rated linear languages}}
\put(165,-4){\makebox(0,0){\it Regular languages}}
\normalsize
\end{picture}
\caption{The hierarchy of some context-free language classes}\label{hier}
\end{figure}
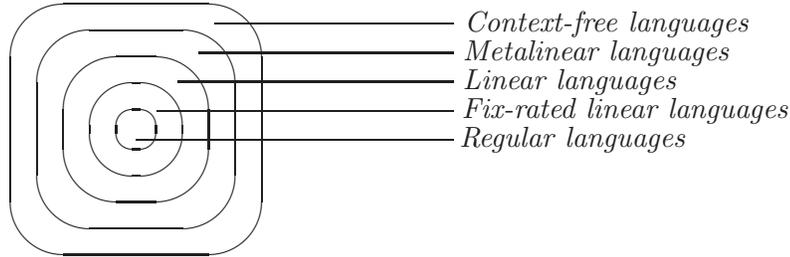

\noindent
Further, when we consider a special fixed value of $k$, then we will also use it as $k = \frac{g}{h}$, where $ g, h \in \mathbb N$ ($h\ne 0$) are relatively primes.

Now we present normal forms for the rules of linear, $k$-rated linear and so,
even-linear
and  regular grammars. 

The following fact is well-known:
\noindent Every linear grammar has an equivalent grammar in which all rules
are in forms of $A \rightarrow aB, A \rightarrow Ba, A \rightarrow
a$ with $a\in V, A,B\in N$. 

\begin{lemma}[Normal form for $k$-rated linear grammars]\label{k-l-nf}
Every $k$-rated \linebreak ($k=\frac{g}{h}$) linear grammar has an equivalent one in which
 for every rule of the form $A \to vBw$: $|w|=g$ and $|v|=h$ such that $g$ and $h$ are relatively primes and
 for all rules of the form $A \to u$ with $u\in V^*$: $|u| < g+h$ 
 holds.
\end{lemma}
\begin{proof}
It goes in the standard way: longer rules can be simulated by shorter ones by the help of newly introduced nonterminals.
\end{proof}

As special cases of the previous lemma we have:
\begin{remark} Every even-linear grammar has an equivalent grammar in
which all rules are in forms $A \rightarrow
aBb, A \rightarrow a$, $A \rightarrow \lambda$ ($A,B \in N, a,b\in V$).
\end{remark}
\begin{remark} Every regular language can be generated by grammar having only
rules of types $A \rightarrow aB , 
 A\rightarrow \lambda$ ($A,B \in N,
a \in V$).
\end{remark}

\newpage


Derivation trees are widely used graphical representations of derivations in context-free grammars. The root of the tree is a node labelled by the initial symbol $S$.
The terminal labelled nodes are leaves of the tree. The nonterminals, as the derivation 
continues from them,
have some children nodes.
Since there is a grammar in Chomsky normal form for every context-free grammar, every word
of a context-free language can be generated such that its derivation tree is a binary
tree.

In linear case, there is at most one non-terminal in every level of the tree. Therefore
the derivation can go only in a linear (sequential) manner. There is only one main branch of the derivation (tree); all the other branches terminate immediately.
Observing the derivations and derivation trees for linear grammars, they seem to be
highly related to the regular case. The linear (and so, specially, the even-linear and fixed linear) languages can be accepted by finite state
machines \cite{am-put1,russanka,nagyB}. Moreover the $k$-rated linear
languages are accepted by deterministic machines \cite{nagyB}.

By an analysis of the possible trees and iterations of nonterminals in a derivation (tree) one can obtain 
pumping (or iteration) lemmas.

Further in this section we recall some well-known iteration lemmas.

\noindent The most famous iteration lemma works for every context-free languages \cite{BARH}.
\begin{lemma}[Bar-Hillel lemma]
\label{BH}
Let a context-free language $L$ be given. Then there exists an integer $n \in\mathbb N$
such that any word $p\in L$ with $|p|\geq n$, admits a factorization
$p=uvwxy$ satisfying

1. $uv^iwx^iy\in L$ for all $i\in\mathbb N$

2. $|vx|>0$

3. $|vwx|\leq n$.
\end{lemma}

\begin{example}
Let $L=\{a^ib^ic^i~|~i\in\mathbb N\}$. It is easy
to show with the Bar-Hillel lemma that the language $L$ is not context-free.
\end{example}

The next lemma works for linear languages \cite{hopul}.
\begin{lemma}[Pumping lemma for linear languages]
\label{LI}
Let $L$ be
 a linear language. Then there exists an integer $n$
such that any word $p\in L$ with $|p|\geq n$, admits a factorization
$p=uvwxy$ satisfying

1. $uv^iwx^iy\in L$ for all integer $i\in\mathbb N$

2. $|vx|>0$

3. $|uvxy|\leq n$.
\end{lemma}

\begin{example}
It is easy to show by using Lemma \ref{LI} that the language \\
$L=\{a^ib^ic^jd^j | i,j\in\mathbb N\}$ is not linear.
\end{example}

In \cite{Geza} there is a pumping lemma for non-linear context-free languages
that can also be effectively used for some languages.
\begin{lemma}[Pumping lemma for non-linear context-free languages]
\label{GE}
Let $L$ be a non-linear context-free language. Then there exist infinite
many words $p\in L$ which admit a factorization $p=rstuvwxyz$
satisfying

1. $rs^itu^ivw^jxy^jz\in L$ for all integer $i,j\geq0$

2. $|su|\neq0$

3. $|wy|\neq0$.
\end{lemma}

\begin{example}
Let $$H\subseteq\{1^2,2^2,3^2,\ldots \}$$ be an infinite set, and let
$$L_H=\{a^kb^ka^lb^l\}~|~
k,l\geq 1;~k\in H~or~l\in H\} \cup \{a^mb^m~|~m\geq1\}.$$

The language $L_H$ satisfies the Bar-Hillel condition. Therefore we can not
apply the Bar-Hillel Lemma to show that $L_H$ is not context-free.
However the $L_H$ language does not satisfy the condition of the
pumping lemma for linear languages. Thus $L_H$ is not linear.
At this point we can apply Lemma \ref{GE}, and the
language  $L_H$ does not satisfy its condition. This means $L_H$ is not context-free.
\end{example}

Now we recall the well-known iteration lemma for regular case (see, for instance, \cite{hopul}).
\begin{lemma}[Pumping lemma for regular languages]\label{reg-pump}
Let $L$ be
 a regular language. Then there exists an integer $n$
such that any word $p\in L$ with $|p|\geq n$, admits a factorization
$p=uvw$ satisfying

1. $uv^iw\in L$ for all integer $i\in\mathbb N$

2. $|v|>0$

3. $|uv|\leq n$.

\end{lemma}

\begin{example}
By the previous lemma one can easily show that the language $\{a^nb^n| n\in \mathbb N\}$
is not regular.
\end{example}

Pumping lemmas are strongly connected to derivation trees, therefore
they works for context-free languages (and for some special subclasses of the context-free languages).

In the next section 
 we present pumping lemmas for the $k$-rated linear languages and for the not metalinear context-free languages.

\section{Main results}

Let us consider a $k$-rated linear grammar.
Based on the normal form (Lemma \ref{k-l-nf}) every word of a $k = \frac{g}{h}$-rated linear
language can be generated by a %
 `pine-tree' shape derivation tree (see Fig. \ref{k-tree}).

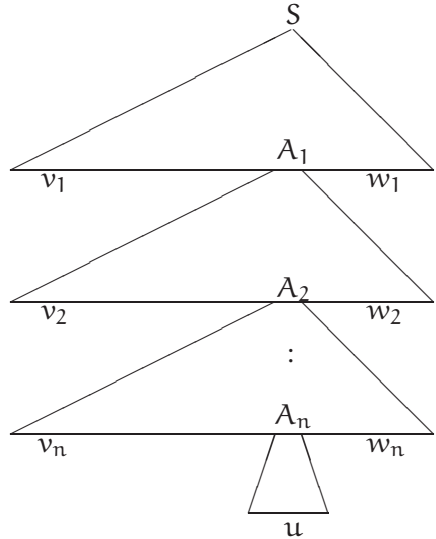
\begin{figure}[httb]
\begin{picture}(300,200)(0,0)
\put (40,40){\line(2,1){100}}
\put (40,90){\line(2,1){100}}
\put (40,90){\line(1,0){160}}
\put (40,140){\line(1,0){160}}
\put (40,140){\line(2,1){106}}
\put (40,40){\line(1,0){160}}
\put (200,40){\line(-1,1){50}}
\put (200,90){\line(-1,1){50}}
\put (200,140){\line(-1,1){53}}
\put (130,10){\line(1,3){10}}
\put (130,10){\line(1,0){30}}
\put (160,10){\line(-1,3){10}}
\put(120,85){\makebox(0,0){$v_2$ \qquad \qquad \qquad \qquad \qquad $w_2$}}
\put(120,35){\makebox(0,0){$v_n$ \qquad \qquad \qquad \qquad \qquad $w_n$}}
\put(120,135){\makebox(0,0){$v_1$ \qquad \qquad \qquad \qquad \qquad $w_1$}}
\put(147,4){\makebox(0,0){$ u $}}
\put(147,199){\makebox(0,0){$ S $}}
\put(147,147){\makebox(0,0){$ A_1 $}}
\put(147,95){\makebox(0,0){$ A_2 $}}
\put(146,70){\makebox(0,0){$ : $}}
\put(147,47){\makebox(0,0){$ A_n $}}
\normalsize
\end{picture}
\caption{A `pine-tree' shape derivation tree in a fix-rated linear grammar}
\label{k-tree}\end{figure}
\noindent


Now we are ready to present our pumping lemmas for these languages.

\begin{theorem}\label{THM-1}
Let $L$ be
 a $(\frac{g}{h}=k)$-rated linear language. Then there exists an integer $n$
such that any word $p\in L$ with $|p|\geq n$, admits a factorization
$p=uvwxy$ satisfying

1. $uv^iwx^iy\in L$ for all integer $i\in\mathbb N$

2. $0< |u|, |v| \leq n \frac{h}{g+h}$

3. $0< |x|, |y| \leq n \frac{g}{g+h}$ 

4. $\frac{|x|}{|v|} = \frac{|y|}{|u|} = \frac{g}{h} = k$.
\end{theorem}
\begin{proof}
Let $G = (N,V,S,H)$ be a $k$-rated linear grammar in normal form that
generates the language $L$. Then let $n= (|N|+1) \cdot (g+h)$.
In this way any word $p$ with length at least $n$ cannot be generated without
any repetition of a nonterminal 
 in the sentential form.
Moreover, by the pigeonhole principle, there is a nonterminal in the derivation which occurs in the sentential forms during the first $|N|$ steps 
 of the derivation and after the first occurrence it occurs also in the next $|N|$ 
sentential forms.
Considering the first two occurrences of this nonterminal $A$ in the derivation tree, the word $p$ can be partitioned to five parts in the following way.
Let $u$ and $y$ be the prefix and suffix (respectively) generated by the first steps till the first occurrence of $A$.
Let $v$ and $x$ be the subwords that are generated from the first occurrence of $A$ till it appears secondly in the sentential form. Finally let $w$ be the
subword that is generated from the second occurrence of $A$ in the derivation.
 (See also Fig. \ref{pump1}.)
 In this way the conditions 2, 3 and 4 of the theorem are fulfilled for the lengths of the partitions.
 Now let us consider the derivation steps between the first two occurrences of $A$.
 They can be omitted from the derivation; in this way the word $uwy$ is obtained.
 This sequence of steps can also be repeated any time, in this way the words of the form
 $uv^iwx^iy$ are obtained for any $i\in\mathbb N$. Thus the theorem is proved.
\end{proof}

\begin{figure}
 \begin{picture}(300,150)(0,30)
\put (40,141){\line(4,1){102}}
\put (200,141){\line(-2,1){50}}
\put (40,126){\line(1,0){100}}
\put (150,126){\line(1,0){50}}
\put (40,126){\line(0,1){15}}
\put (200,126){\line(0,1){15}}
\put (40,100){\line(4,1){102}}
\put (200,100){\line(-2,1){50}}
\put (200,60){\line(-2,1){50}}
\put (40,40){\line(1,0){160}}
\put (40,60){\line(4,1){102}}
\put (40,40){\line(0,1){20}}
\put (200,40){\line(0,1){20}}
\put (40,85){\line(0,1){15}}
\put (200,85){\line(0,1){15}}
\put (40,85){\line(1,0){100}}
\put (150,85){\line(1,0){50}}
\put(120,80){\makebox(0,0){$v$ \qquad \qquad \qquad \qquad \qquad \qquad \quad $x$}}
\put(132,33){\makebox(0,0){$w$}}
\put(120,120){\makebox(0,0){$u$ \qquad \qquad \qquad \qquad \qquad \qquad \quad $y$}}
\put(146,168){\makebox(0,0){$ S $}}
\put(146,87){\makebox(0,0){$ A $}}
\put(146,127){\makebox(0,0){$ A $}}
\put(147,64){\makebox(0,0){$ : $}}
\put(147,107){\makebox(0,0){$ : $}}
\put(147,147){\makebox(0,0){$ : $}}
\normalsize
\end{picture}
  \caption{Pumping the subwords between the two occurrences of the non-terminal $A$.}\label{pump1}
\end{figure}
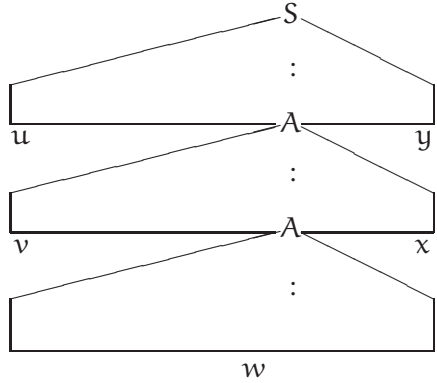

\begin{theorem}\label{THM-2}
Let $L$ be
 a $(\frac{g}{h}=k)$-rated linear language. Then there exists an integer $n$
such that any word $p\in L$ with $|p|\geq n$, admits a factorization
$p=uvwxy$ satisfying

1. $uv^iwx^iy\in L$ for all integer $i\in\mathbb N$

2. $0< |v| \leq n \frac{h}{g+h}$ 

3. $0< |x| \leq n \frac{g}{g+h}$

4. $0< |w| \leq n $

5. $\frac{|x|}{|v|} = \frac{|y|}{|u|} = \frac{g}{h} = k$.
\end{theorem}
\begin{proof}
Let $G = (N,V,S,H)$ be a $k$-rated linear grammar in normal form that
generates the language $L$. Then let $n= (|N|+1) \cdot (g+h)$.
In this way any word $p$ with length at least $n$ cannot be generated without
any repetition of a nonterminal 
 in the sentential form.
Moreover there is a nonterminal $A$ in the derivation which occurs twice among the non-terminals of the last $|N+1|$ sentential
forms of the derivation. 
Considering these last two occurrences of 
 $A$ in the derivation tree the word $p$ can be partitioned to five parts in the following way.
Let $u$ and $y$ be the prefix and suffix (respectively) generated from the first steps till that occurrence
of $A$ which is the last but one during the derivation. Let $v$ and $x$ be the subwords that are generated by the steps between 
 the last two occurrences of $A$. 
Finally let $w$ be the
subword that is generated from the last occurrence of $A$ in the derivation.
 In this way the conditions 2, 3, 4 and 5 are fulfilled for the lengths of the partitions.
 Now let us consider the derivation steps between the these two occurrences of $A$.
 They can be omitted from the derivation; in this way the word $uwy$ is obtained.
 This sequence of steps can also be repeated any time, in this way the words of the form
 $uv^iwx^iy$ are obtained for any $i\in\mathbb N$. Thus the theorem is proved.
\end{proof}

\begin{remark}
In case of
 $k=0$ 
the previous theorems give the well-known pumping lemmas for regular languages.
\end{remark}

Now we are presenting an iteration lemma for another special subclass of the
context-free language family.

\begin{theorem}
\label{UJ}
Let $L$ be a context-free language which does not belong to
any $k$-linear language for a given positive integer $k$. Then there exist infinite
many words $w\in L$ which admit a factorization
$w=uv_0w_0x_0y_0\ldots v_kw_kx_ky_k$
satisfying

1. $uv_0^{i_0}w_0x_0^{i_0}y_0\ldots v_k^{i_k}w_kx_k^{i_k}y_k\in L$
for all integer $i_0,\ldots ,i_k\geq0$

2. $|v_jx_j|\neq0$ for all $0\leq j\leq k$.
\end{theorem}

\begin{proof} Let $G=(N,V,S,H)$ be a context-free grammar such
that $L(G)=L$, and let $G_A=(N,V,A,H)$ for all $A\in N$.
Because $L$ is not k-linear, there exists $A_0,\ldots ,A_k\in V_N$ and $\alpha,\beta_0,\ldots ,\beta_k\in V^*$ such that
$S\Rightarrow^*\alpha A_0\beta_0\ldots  A_k\beta_k$,
where all of the languages $L(G_{A_l})$, $0\leq l \leq k$ are infinite.
Then the words

$\{\alpha\} L(G_{A_0})\{\beta_0\}\ldots L(G_{A_k})\{\beta_k\}\subseteq L$,
and applying the Bar-Hillel Lemma for all $L(G_{A_l})$ we receive
$\alpha a_0b_0^{i_0}c_0d_0^{i_0}e_0\beta_0\ldots a_kb_k^{i_k}c_kd_k^{i_k}e_k\beta_k\subseteq L$ for all $i_0\geq0,\ldots ,i_k\geq0$.
Let $u=\alpha a_0,~v_l=b_l,~w_l=c_l,~x_l=d_l,~y_l=e_l\beta_l$,
and we have the above form.
\end{proof}

\newpage

\begin{remark} With $k=1$ we have a pumping lemma for
non-linear context-free languages.
\end{remark}

Knowing that every $k$-linear language is metalinear for any $k\in\mathbb N$, we have:

\begin{proposition}
\label{UJ2}
 Let $L$ be a not metalinear context-free language.
 For all integers $k\geq 1$ there exist infinite many words
$w\in L$ which admit a factorization

\noindent
$w=uv_0w_0x_0y_0\ldots v_kw_kx_ky_k$
satisfying

1. $uv_0^{i_0}w_0x_0^{i_0}y_0\ldots v_k^{i_k}w_kx_k^{i_k}y_k\in L$
for all integer $i_0,\ldots ,i_k\geq0$

2. $|v_jx_j|\neq0$ for all $0\leq j\leq k$.
\end{proposition}


\section{Applications of the new iteration lemmas}

As pumping lemmas are usually used to show that a language does not belong to a language class, we present an example for this type of application.

\begin{example}
The DYCK language (the language of correct bracket expressions) is not $k$-linear
for any value of $k$ over the alphabet $\{ (,)\}$. Let $k\ne 1$ be fixed as $\frac{g}{h}$.
Let us consider the word of the form $(^{(g+h)(n+2)})^{(g+h)(n+2)}$. Then Theorem \ref{THM-1}
does not work (if $k\ne 1$), the pumping deletes or introduces different number of $($'s and $)$'s.
To show that the DYCK language is not 1-rated (i.e., even-)linear let us consider
the word $(^{2n})^{2n}(^{2n})^{2n}$. Using Theorem \ref{THM-2} the number of inner brackets
can be pumped. In this way such words are obtained in which there are prefixes with more letters $)$ than $($. Since these words do not belong to the language, this language is not $k$-linear.
\end{example}

In the previous example we showed that the DYCK language is not fixed linear.

In the next example we consider a deterministic linear language.
\begin{example}\label{ex-lin}
Let  $L=\{a^mb^m | m \in \mathbb N \}   \cup \{ a^mcb^{2m} | m\in \mathbb N \}$ over the alphabet $\{a,b,c \}$.
 Let us assume that the language is fixed linear. First we show that this language is not fixed linear with ratio other than 1. On the contrary, assume that it is, with
 $k = \frac{g}{h} \in \mathbb Q$ such that $k\ne 1$. Let $n$ be given by Theorem \ref{THM-1}.
Then consider the words of the form $a^{m(g+h)}b^{m(g+h)}$ with $m>n$. By the theorem any of them can be factorized to $uvwxy$ such that $|uv| \leq \frac{2nh}{g+h}$. Since $g+h > 2$
(remember that $g,h \in\mathbb N$, relatively primes and $g\ne h$), $|uv| < nh$, and
therefore both $u$ and $v$ contains only $a$'s. By a similar argument on the length of $xy$,
$x$ and $y$ contains only $b$'s. Since the ratio $\frac{|x|}{|v|}$ (it is fixed by the theorem) is not 1, by pumping we get words outside of the language.
Now we show that this language is not even-linear.
Assume that it is 1-rated linear ($g=h=1$). Let $n$ be the value from Theorem \ref{THM-1}.
Let us consider the words of shape  $a^{m}cb^{2m}$ with $m>n$. Now we can factorize
these words in a way, that $|uv| \leq {n}$ and $|xy| \leq n$ and $|v| = |x|$. By pumping we get words $a^{m+j}cb^{2m+j}$ with some positive values of $j$, but they are not in $L$.
We have a contradiction again. So this language is not fixed linear.
\end{example}

In the next example we show a fixed-linear language that can be pumped.
\begin{example}
Let $L$ be the language of palindromes, i.e., of the words  
 over $\{a,b \}$ that are the same in reverse order ($p=p^R$). We show that our pumping lemmas 
 work for this language with the value $k=1$.
 Let $p \in L$, then $p = uvwxy$ according to Theorem \ref{THM-1} or Theorem \ref{THM-2},
 such that $|u|=|y|$ and $|v|=|x|$. Therefore, by applying the main property of the palindromes, we have $u = y^R$, $v = x^R$ and $w=w^R$.
 By $i=0$ the word $uwy$ is obtained which is in $L$ according to the previous equalities.
 By further pumping the words $u v^i w x^i y$ are obtained, they are also palindromes.
To show that this language cannot be pumped with any other values, let us consider
words of shape $a^m b a^m$. By Theorem \ref{THM-1} it can be shown in analogous way that
we showed in Example \ref{ex-lin} that enough long words cannot be pumped with ratio $k\ne 1$.
\end{example}

Besides our theorems work for regular languages with $k=0$ there is a non-standard application of them.
As we already mentioned, all regular languages are $k$-rated linear for any values of $k\in\mathbb Q$. Therefore every new pumping lemma works for any regular language with
any values of $k$. Now we show some examples.

\begin{example}
Let the regular language $(ab)^*aa(bbb)^*a$ be given.
Then we show, that our theorems work for, let us say, $k=\frac{1}{2}$. 
Every word of the language is of the form $(ab)^naa(bbb)^ma$ (with $n,m\in\mathbb N$). For  words that are long enough either
$n$ or $m$ (or both of them) are sufficiently large. Now we detail effective factorizations
$p=uvwxy$ of the possible cases. We give only those words of the factorization that have maximized lengths due to the applied theorem, the other words can easily be found by the factorization and, at Theorem \ref{THM-2}, by taking into account the fixed ratio of some lengths in the factorization.
{\itemize
\item

Theorem \ref{THM-1} for $k=\frac{1}{2}$: \\
if $n>3$ and $m>0$ : let $u=ab$, $v=ababab$, $x=bbb$, $y=a$, \\
if $m = 0$ : let $u=ababab$, $v=abab$, $x=ab$, $y=aaa$, \\
if $n=3$ : let $u=abababaa$, $v=bb$, $x=b$, $y=bbba$,\\
if $n=2$ : let $u=ababaa$, $v=bb$, $x=b$, $y=bba$,\\
if $n=1$ : let $u=abaa$, $v=bb$, $x=b$, $y=ba$,\\
if $n=0$ : let $u=aa$, $v=bb$, $x=b$, $y=a$.
\item

Theorem \ref{THM-2} for $k=\frac{1}{2}$: \\
if $n\leq 3m-4$ : let $v=bb$, $w=b $ $x=b$, \\
if $n=3m-3$ : let $v=ababab$, $w=aabbbb $ $x=bbb$, \\
if $n=3m-2$ : let $v=ababab$, $w=abaabbbb $ $x=bbb$, \\
if $n=3m-1$ : let $v=ababab$, $w=ababaabbbb $ $x=bbb$, \\
if $n=3m$ : let $v=ababab$, $w=aab $ $x=bbb$, \\
if $n = 3m+1$ : let $v=ababab$, $w=abaab$, $x=bbb$, \\
if $n = 3m+2$ : let $v=ababab$, $w=ababaab$, $x=bbb$, \\
if $n = 3m+3$ : let $v=ababab$, $w=abababaab$, $x=bbb$, \\
if $n = 3m+4$ : let $v=ababab$, $w=ababababaab$, $x=bbb$, \\
if $n = 3m+5$ : let $v=ababab$, $w=abababababaab$, $x=bbb$, \\
if $n \geq 3m+6$, $n\equiv0 ($mod$ 3)$ : let $v=abab$, $w=\lambda$, $x=ab$, \\
if $n \geq 3m+7$, $n\equiv1 ($mod$ 3)$ : let $v=abab$, $w=ab$, $x=ab$, \\
if $n \geq 3m+8$, $n\equiv2 ($mod$ 3)$ : let $v=abab$, $w=abab$, $x=ab$. \\

}
\end{example}
In similar way it can be shown  that pumping the words of a regular language in two places simultaneously  with other values of $k$ (for instance, $1, 5, \frac{7}{3}$ etc.) works.

In the next example we show that there are languages that can be pumped by the usual
pumping lemmas for regular languages, but they cannot be regular since we prove 
that there is a value of $k$ such that one of our theorems does not work.

\begin{example}
Let $L = \{a^r b a^q b^m | r,q,m \geq 2, \exists j\in\mathbb N: q=j^2 \}$.
By the usual pumping lemmas for regular languages, i.e., by fixing $k$ as 0,
one cannot infer that this language is not regular. By $k=0$, $x=y=\lambda$ and so
$p = uvw$.
Due to the $a$'s in the beginning, Theorem \ref{THM-1} works: $u=a,v=a$; 
and due to the $b$'s in the end Theorem \ref{THM-2} also works: $v=b,w=b$. \\
Now we show that $L$ is not even-linear.
Contrary, let us assume that Theorem \ref{THM-2} works for $k=1$. Let $n$ be the value
for this language according to the theorem.
Let $p = a^2 b a^{(2n+5)^2} b^3$. By the conditions of the theorem, it can be factorized
to $uvwxy$ such that $|v|,|w|,|x| \leq n$ and $|u| = |y|$.
In this way $vwx$ must be a subword of $a^{(2n+5)^2}$, and so, the pumping decreases/increases
only $q$. Since $|v| , |x| \leq n$ in the first round of pumping
$p' = a^2 b a^{(2n+5)^2 + |vx|} b^3$ is obtained. But $(2n+5)^2 < (2n+5)^2 + |vx| \leq
(2n+5)^2 + 2n < (2n+6)^2$, therefore $p' \not\in L$. \\
 Thus $L$ is not even-linear, and therefore it cannot be regular. Our pumping lemma was effective to show this fact.
\end{example}

Usually pumping lemmas can be used only to show that some languages do not belong
 to the given class of languages. One may ask what we can say if a language satisfy
 our theorems.
Now we present an example which shows that we cannot infer about the language class
 if a language satisfies our new pumping lemmas.

\begin{example}
Let $L = \{ 0^j1^m0^r1^i 0^l 1^i 0^r 1^m 0^j | j,m,i,l,r \geq 1, r \textnormal { is prime} \}$.
One can easily show that this language satisfies both Theorem \ref{THM-1}
and Theorem \ref{THM-2} with $k=1$: one can find subwords to pump in the part of outer $0$'s or $1$'s (pumping their number form a given $j$ or $m$ to arbitrary high values), or in the middle part $0$'s or $1$'s (pumping their number from $i$ or $l$ to arbitrary high values), respectively.
But this language is not even context-free, since intersected by the regular language
$010^*1010^*10$ a non semi-linear language is obtained. Since context-free languages are semi-linear (due to the Parikh theorem) and the class of context-free languages are closed
under intersection with regular languages, we just proved that $L$ cannot be linear or
fix-rated linear.
\end{example}

It is a more interesting question what we can say about a language for which there
are values $k_1 \ne k_2$ such that all its enough long words can be pumped
both as $k_1$-rated and $k_2$-rated linear language.
We have the following conjecture.
\begin{conjecture}
If a language $L$ satisfies any of our pumping lemmas for two different values of $k$,
then $L$ is regular.
\end{conjecture}
If the previous conjecture is true, then exactly the regular languages  form
the intersection of the $k$-rated linear language families (for $k\in \mathbb Q$).

Regarding iteration lemma for the not metalinear case, we show two examples.

\begin{example}
This is a very simple example, we can use our lemma to show that the language
$$L_1=\{a^lb^la^mb^ma^nb^n~|~l,m,n\geq 0\}$$ is metalinear.

First of all, it is easy to show that $L_1$ is context-free. The language $L_1$ does
not satisfy the condition of the pumping lemma for not metalinear context-free
languages, (Proposition \ref{UJ2},) so $L_1$ must be a metalinear context-free language.
\end{example}

In our next example we show a more complicated language which satisfies the Bar-Hillel
condition, and we use our pumping lemma to show that the language is not context-free.

\begin{example}
Let $$H\subseteq\{2^k~|~k\in\mathbb N\}$$ be an infinite set, and let
$$L_2=\{a^lb^la^mb^ma^nb^n~|~
l,m,n\geq 1;~l\in H~or~m\in H~or~n\in H\} \cup$$
$$\cup \{a^ib^ia^jb^j~|~i,j\geq1\}.$$

$L_2$ satisfies the Bar-Hillel condition. Therefore we can not
apply the Bar-Hillel Lemma to show that $L_2$ is not context-free.
However it is easy to show that $L_2$ is not $3$-linear language.
Now we can apply Theorem \ref{UJ}, and the language  $L_2$ does
not satisfy its condition with $k=3$. This means $L_2$ does not belong
to the not $3$-linear context-free languages, so the language $L_2$ is
not context-free.
\end{example}

\section{Conclusions}

In this paper some new pumping lemmas are proved for special context-free and linear languages. In fix-rated
linear languages the lengths of the pumped subwords of a word depend on each other,
therefore these pumping lemmas are more restricted than the ones working on every linear
or every context-free languages.
Since all regular languages are $k$-rated linear for any non-negative rational value of $k$, 
these lemmas also work for regular languages.
The question whether only regular languages satisfy our pumping lemmas at least for two
different values of $k$ (or for all values of $k$) is remained open as a conjecture.
We also investigated a special subclass of context-free language family and
introduced iteration conditions which is satisfied only not metalinear context-free languages.
These conditions can be used in two different ways. First they can be used to proove that a
language is not context-free. On the other hand, we can also use them to show that the given
language is belong to the metalinear language family.

\begin{figure}
\begin{picture}(300,100)(-70,-52)
\put (0,0){\oval(95,95)}
\put (0,0){\oval(75,75)}
\put (0,0){\oval(55,55)}
\put (0,0){\oval(35,35)}
\put (30,40){\line(1,0){90}}
\put (24,29){\line(1,0){96}}
\put (16,18){\line(1,0){104}}
\put (8,7){\line(1,0){112}}
\put (-37,15){\line(1,1){22}}
\put (-37,-15){\line(1,1){10}}
\put (5,27){\line(1,1){10}}
\put (-37,0){\line(1,1){10}}
\put (-10,27){\line(1,1){10}}
\put (-35,-28){\line(1,1){10}}
\put (18,25){\line(1,1){10}}
\put (-27,-35){\line(1,1){10}}
\put (25,17){\line(1,1){10}}
\put (-14,-37){\line(1,1){10}}
\put (28,5){\line(1,1){10}}
\put (1,-37){\line(1,1){10}}
\put (28,-10){\line(1,1){10}}
\put (16,-37){\line(1,1){22}}
\put (6,-17){\line(1,1){11}}
\put (-8,-16){\line(1,1){24}}
\put (-16,-9){\line(1,1){25}}
\put (-17,5){\line(1,1){12}}
\put(192,40){\makebox(0,0){\it Context-sensitive languages}}
\put(179,29){\makebox(0,0){\it Context-free languages}}
\put(175,18){\makebox(0,0){\it Metalinear languages}}
\put(187,7){\makebox(0,0){\it Fix-rated linear languages}}
\normalsize
\end{picture}
\caption{The target language classes of the new iteration lemmas}
\label{hier}\end{figure}
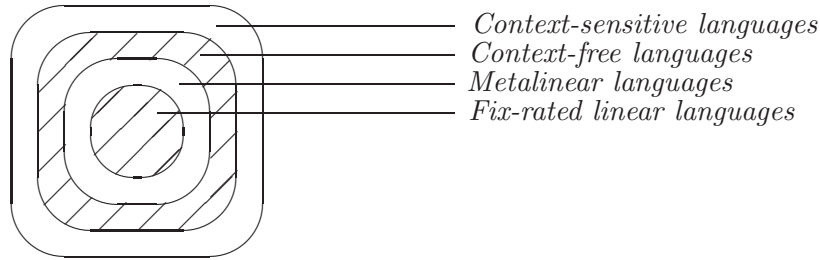
\noindent

\section*{Acknowledgements}
{The work is supported by the Czech-Hungarian bilateral project (T\'eT)
and the T\'AMOP 4.2.1/B-09/1/KONV-2010-0007 project.
The project is implemented through the New Hungary Development Plan,
co-financed by the European Social Fund and the European Regional
Development Fund.}

\bigskip
\rightline{\emph{Received: October 5, 2010 {\tiny \raisebox{2pt}{$\bullet$\!}} Revised:	November 2, 2010}}     


\begin{thebibliography}{99}

\bibitem{am-put1} V. Amar, G. R. Putzolu, On a family of linear grammars, \emph{Information and
\href{http://www.sciencedirect.com/science/journal/00199958}{Control}}, \textbf{7,} 3 (1964) 283--291.


\bibitem{am-put2} V. Amar, G. R.  Putzolu,  Generalizations of regular events, \emph{Information and
\href{http://www.sciencedirect.com/science/journal/00199958}{Control}}, \textbf{8,} 1 (1965) 56--63.


\bibitem{BARH}
Y. \href{http://en.wikipedia.org/wiki/Yehoshua_Bar-Hillel}{Bar-Hillel}, M. Perles,  and E. Shamir, On formal properties of simple phrase structure grammars, \emph{Z. Phonetik.
Sprachwiss. Komm.}, \textbf{14} (1961) 143--172.


\bibitem{pali-pump} P. \href{http://www.inf.unideb.hu/~domosi/}{D\"om\"osi}, M. Ito, M. Katsura, C. Nehaniv,  New pumping property of context-free languages, \emph{Combinatorics, Complexity an Logic, Proc. International Conference on Discrete Mathemtics and Theoretical Computer Science --  DMTCS'96},   \href{http://www.springer.com/}{Springer}, Singapore, pp. 187--193.

\bibitem{hopul}
J. E. \href{http://www.cs.cornell.edu/Info/Department/Annual95/Faculty/Hopcroft.html}{Hopcroft}, 
J. D. \href{http://infolab.stanford.edu/~ullman/}{Ullman}, \emph{Introduction to automata theory,
languages, and computation,} (2nd edition),  \href{http://www.pearsonhighered.com}{Addison-Wesley}, Reading, MA,
1979. 


\bibitem{Geza}
G. \href{http://www.inf.unideb.hu/~geza/}{Horv\'{a}th}, New pumping lemma for non-linear
context-free languages, \emph{Proc. 9th Symposium on Algebras, Languages and Computation},
Shimane University, Matsue, Japan, 2006, pp. 160--163.

\bibitem{russanka} R. Lokunova, Linear context free languages, \emph{Proc. ICTAC 2007, Lecture Notes in }  
\href{http://www.springer.com/series/558}{\emph{Comput. Sci.,}} \textbf{4711} (2007) 351--365.

\bibitem{nagyB} B. \href{http://www.inf.unideb.hu/~nbenedek/}{Nagy},  On $5'\to 3'$ sensing Watson-Crick finite automata, \emph{DNA 13, Revised selected papers, Lecture Notes in } \href{http://www.springer.com/series/558}{\emph{Comput. Sci.,}} \textbf{4848} (2008) 256--262.

\bibitem{handbo}
G. \href{http://www.liacs.nl/~rozenber/}{Rozenberg}, A. \href{http://en.wikipedia.org/wiki/Arto_Salomaa}{Salomaa}, (eds.) \emph{Handbook of formal languages}, \href{http://www.springer.com/}{Springer}, Berlin, Heidelberg,  1997.

\bibitem{even-lin} J. M. Sempere, P. Garc\'\i a, 
A characterization of even linear languages and its application to
the learning problem, 
 \emph{Proc. Second International Colloquium, ICGI-94,
Lecture Notes in } 
\href{http://www.springer.com/series/1244}{\emph{Artificial Intelligence,}} 
\textbf{862} (1994) 38--44. 

\end{thebibliography}
\end{document}